\documentclass[12pt]{article}

\usepackage{fullpage,cmu-titlepage2}

\ifx\pdftexversion\undefined
  \usepackage[dvips]{graphics}
\else
  \usepackage[pdftex]{graphics}
\fi

\usepackage{rotating}
\usepackage{tabularx}
\usepackage[
pageanchor=true,
plainpages=false
pdfpagelabels, 
bookmarks,
bookmarksnumbered,
pdfborder=0 0 0,
pdfpagemode=UseOutlines]{hyperref}
\usepackage{subfigure}
\usepackage{latexsym,amsthm,amssymb,amscd,url,enumerate}
\usepackage{amsmath}
\usepackage{mathpartir,mathtools}
\usepackage[all]{xy}
\usepackage{tikz}
\usepackage{tikz-cd}
\usetikzlibrary{arrows}

\newcommand{\T}{\mathcal{T}}
\newcommand{\I}{\mathcal{I}}
\newcommand{\baseT}{\mathtt{b}}
\newcommand{\pT}{\mathtt{p}}
\newcommand{\qT}{\mathtt{q}}
\newcommand{\tT}{\mathtt{t}}
\newcommand{\type}{\mathtt{type}}
\newcommand{\id}[3][]{#2 =_{#1} #3}
\newcommand{\idtype}[3][]{\id[#1]{#2}{#3}}
\newcommand{\refl}[1]{\ensuremath{\mathsf{1}_{#1}}}
\newcommand{\defeq}{\coloneqq}
\newcommand{\sm}[1]{\Sigma_{#1}}
\newcommand{\prd}[1]{\Pi_{#1}}
\newcommand{\lam}[1]{\lambda_{#1}}

\newcommand{\iseq}[1]{\mathsf{iseq}(#1)}

\newcommand{\nathom}{\mathsf{nat}}

\newcommand{\ap}[2]{\mathsf{ap}_{#1}(#2)}
\newcommand{\J}[2]{\mathsf{J}(#1,#2)}
\newcommand{\opp}[1]{\mathord{{#1}^{-1}}}

\newcommand{\idfun}[1]{\mathsf{id}_{#1}}
\newcommand{\comp}{\circ}

\newcommand{\pair}{\ensuremath{\mathsf{proj}^{\mathord{=}}}}
\newcommand{\pairpath}{\ensuremath{\mathsf{pair}^{\mathord{=}}}}
\newcommand{\funext}{\mathsf{funext}}
\newcommand{\happly}{\mathsf{hap}}

\newcommand{\fst}{\pi_1}
\newcommand{\snd}{\pi_2}

\newcommand{\Sn}{\mathbf{S}}
\newcommand{\base}{\mathsf{base}}
\newcommand{\lloop}{\mathsf{loop}}

\newcommand{\transfib}[3]{\ensuremath{\mathsf{trans}^{#1}(#2,#3)}}
\newcommand{\ct}{%
  \mathchoice{\mathbin{\raisebox{0.5ex}{$\displaystyle\centerdot$}}}%
             {\mathbin{\raisebox{0.5ex}{$\centerdot$}}}%
             {\mathbin{\raisebox{0.25ex}{$\scriptstyle\,\centerdot\,$}}}%
             {\mathbin{\raisebox{0.1ex}{$\scriptscriptstyle\,\centerdot\,$}}}}
\newcommand{\mapdepfunc}[1]{\ensuremath{\mathsf{apd}_{#1}}}
\newcommand{\mapdep}[2]{\ensuremath{\mapdepfunc{#1}\mathopen{}\left(#2\right)\mathclose{}}}

\newtheorem{theorem}{Theorem}						
\newtheorem{definition}[theorem]{Definition}

\newtheorem{lemma}[theorem]{Lemma}

\title{The equivalence of the torus and the product of two circles in homotopy type theory}

\author{Kristina Sojakova}

\date{September 2015}

\abstract{Homotopy type theory is a new branch of mathematics which merges insights from abstract homotopy theory and higher category theory with those of logic and type theory. It allows us to represent a variety of mathematical objects as basic type-theoretic constructions, higher inductive types. We present a proof that in homotopy type theory, the torus is equivalent to the product of two circles. This result indicates that the synthetic definition of torus as a higher inductive type is indeed correct.}


\trnumber{CMU-CS-15-105}

\support{I gratefully acknowledge the support of the Air Force Office of Scientific Research through MURI grant FA9550-15-1-0053. Any opinions, findings and conclusions or recommendations expressed in this material are those of the authors and do not necessarily reflect the views of the AFOSR.}

\begin{document}
\renewcommand*{\thepage}{title-\arabic{page}} 
\maketitle
\renewcommand*{\thepage}{\arabic{page}} 

\section{Introduction}

Homotopy type theory (HoTT) \cite{hott} is a new branch of mathematics which merges insights from abstract homotopy theory and higher category theory with those of logic and type theory. A number of well-known results in algebraic topology have been established within HoTT and formally verified using the proof assistants Agda \cite{agda} and Coq \cite{coq}; these include the calculation of $\pi_n(\mathbf{S}^n)$ (\cite{licata_shulman,licata_brunerie_13}); the Freudenthal Suspension Theorem \cite{hott}; the Blakers-Massey Theorem \cite{hott}, the van Kampen theorem \cite{hott}, and the Mayer-Vietoris theorem \cite{mayer-vietoris}. 

As a formal system, HoTT is an extension of Martin-L{\"o}f's dependent type theory with two new concepts: Voevodsky's \emph{univalence axiom} (\cite{ssets,uf_talk}) and \emph{higher inductive types} (\cite{lumsdaine_blog,shulman_blog}). The univalence axiom can be paraphrased as stating that \emph{equivalent types are equal}, and hence we can reason about them using the identity elimination principle. While we do not make an explicit use of the axiom in this paper, we use one of its most important consequences - the function extensionality principle - which states that two pointwise equal functions are in fact equal (\cite {funext_proof}, Ch.~4.9 of \cite{hott}). 

The second main feature of HoTT, higher inductive types, are a higher-dimensional generalization of ordinary inductive types which allows us to declare constructors involving the \emph{path spaces} of the type $X$ being defined, rather than just $X$ itself. This means that we can define the higher inductive type $X$ \emph{e.g.}, by means of the constructors $\base : X$, $\lloop : \base =_X \base$. While $\base$ is an ordinary nullary constructor, akin to the constant $0$ in the definition of natural numbers, $\lloop$ is a term of an identity type over $X$, not $X$ itself. Intuitively, we can draw the type $X$ as consisting of the point $\base$ and a loop from $\base$ to $\base$ - also known as the circle:

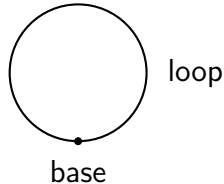
\begin{center}	
\begin{tikzpicture}[scale=1.3,cap=round,>=latex]
        \draw[thick] (0cm,0cm) circle(0.7cm);
				\filldraw[black] (270:0.7cm) circle(1pt);           
				\draw (270:1cm) node[fill=white] {$\base$};				
			  \draw (0:1.2cm) node[fill=white] {$\lloop$};
\end{tikzpicture}			
\end{center}

This is not an isolated occurrence: higher inductive types turn out to be well suited for representing a wide variety of mathematical objects, and the definitions generally require very little prior development. Most of the difficult work then lies in showing that such a ``synthetic" definition is indeed the ``right" one, in the sense that the higher inductive type representing, \emph{e.g.}, the circle or the torus does possess the expected mathematical properties. For instance, we would like to be able to show that in HoTT, the fundamental group of the circle is the group of integers, and that the torus is the product of two circles.

The former result was shown by \cite{licata_shulman} and notably, the proof they give is much more concise than its homotopy-theoretic counterpart. In this paper, we present the full proof of the latter result that the torus $T^2$ is equivalent, in a precise sense, to the product $\Sn^1 \times \Sn^1$ of two circles. This problem was brought to the author's attention during the Special Year on Univalent Foundations at the Institute for Advanced Study in 2012/2013. During that time, the author gave a sketch of the proof\footnote{In personal correspondence, P. Lumsdaine stated he also had a sketch of a proof, which has not been made public.} and a year later expanded it into a full writeup \cite{torus}, which was included in the HoTT Book exercise solutions file but never published. In summer of 2014, Dan Licata and Guillaume Brunerie produced a similar, formalized proof of the result which builds upon their cubical library for the Agda proof assistant. This proof later appeared in a published paper \cite{licata_brunerie}. In the conclusion we provide a more detailed comparison of how the proof presented here compares to the one by Licata and Brunerie. 

In \cite{licata_cubical}, Licata later presented a proof of the same result in cubical type theory. This proof is much simpler since the cubical type theory seems better suited for arguments involving higher paths; however, this new theory is itself still under development.


\section{Preliminaries}
Summarizing from \cite{hott}, HoTT is a dependent type theory with
\begin{itemize}

\item dependent pair types $\sm{x:A}{B(x)}$ and dependent function types $\prd{x:A}{B(x)}$. The non-dependent versions are denoted by $A \times B$ and $A \to B$.

\item intensional identity types $\idtype[A]{x}{y}$. We have the usual formation and introduction rules, where the identity path on $x : A$ will be denoted by $\refl{x}$. The elimination and computation rules are recalled below:
\begin{mathpar}
\inferrule{E : \prd{x,y : A} \idtype[A]{x}{y} \to \type \\ d : \prd{x : A} E(x,x,\refl{x})}
          {\J{E}{d} : \prd{x,y : A}\prd{p : \idtype[A]{x}{y}} E(x,y,p)}

\inferrule{E : \prd{x,y : A} \idtype[A]{x}{y} \to \type \\ d : \prd{x : A} E(x,x,\refl{x}) \\ a : A}
          {\J{E}{d}(a,a,\refl{a}) \equiv d(a) : E(a,a,\refl{a})}
\end{mathpar}
As usual, these rules are applicable in any context $\Gamma$, which we generally omit. If the type $\idtype[A]{x}{y}$ is inhabited, we call $x$ and $y$ \emph{equal}. If we do not care about the specific equality witness, we often simply say that $\id[A]{x}{y}$. A term $p : \id[A]{x}{y}$ will be often called a \emph{path} and the process of applying the identity elimination rule will be referred to as \emph{path induction}. Definitional equality between $x,y:A$ will be denoted as $x \equiv y : A$.
\end{itemize}

Proofs of identity behave much like paths in topological spaces: they can be reversed, concatenated, mapped along functions, etc. Below we summarize a few of these properties:
\begin{itemize}
\item For any path $p : \id[A]{x}{y}$ there is a path $\opp{p} : \id[A]{y}{x}$, and we have $(\refl{x})^{-1} \equiv \refl{x}$.
\item For any paths $p : \id[A]{x}{y}$ and $q : \id[A]{y}{z}$ there is a path $p \ct q : \id[A]{x}{z}$, and we have
$\refl{x} \ct \refl{x} \equiv \refl{x}$.
\item Associativity of composition: for any paths $p : \id[A]{x}{y}$, $q : \id[A]{y}{z}$, $r : \id[A]{z}{u}$ we have
$\id{(p \ct q) \ct r}{p \ct (q \ct r)}$.
\item We have $\refl{x} \ct p = p$ and $p \ct \refl{y} = p$ for any $p : \id[A]{x}{y}$.
\item For any $p : \id[A]{x}{y}$, $q : \id[A]{y}{z}$ we have $p \ct \opp{p} = \refl{x}$, $\opp{p} \ct p = \refl{y}$, and $\opp{(\opp{p})} = p$, $\opp{(p \ct q)} = \opp{q} \ct \opp{p}$.
\item For any $P : A \to \type$ and $p : \id[A]{x}{y}$ there is a function $\mathtt{trans}^P(p) : P(x) \to P(y)$   
called the \emph{transport}. We furthermore have $\mathtt{trans}^P(\refl{x}) \equiv \lam{x:P(x)} x$.
\item We have $\mathtt{trans}^P(p \ct q) = \mathtt{trans}^P(q) \comp \mathtt{trans}^P(p)$ for any $P : A \to \type$ and $p : \id[A]{x}{y}$, $q : \id[A]{y}{z}$.
\item For any function $f : A \to B$ and path $p : \id[A]{x}{y}$, there is a path $\ap{f}{p} : \id[B]{f(x)}{f(y)}$ and we have
$\ap{f}{\refl{x}} \equiv \refl{f(x)}$. 
\item We have $\ap{f}{\opp{p}} = \opp{\ap{f}{p}}$ and $\ap{f}{p \ct q} = \ap{f}{p} \ct \ap{f}{q}$ for any $f : A \to B$ and $p : \id[A]{x}{y}$, $q : \id[A]{y}{z}$.
\item Given a dependent function $f : \prd{x:A} B(x)$ and path $p : \id[A]{x}{y}$, there is a path $\mapdepfunc{f}(p) : \id[B(y)]{\transfib{B}{p}{f(x)}}{f(y)}$ and we have $\mapdepfunc{f}(\refl{x}) \equiv \refl{f(x)}$.
\item All constructs respect propositional equality.
\end{itemize}

\begin{definition}
For $f,g : \prd{x:A} B(x)$, we define the type \[f \sim g \defeq \prd{a:A} \id[B(a)]{(f(a)}{g(a))}\] and call it the \emph{type of homotopies} between $f$ and $g$.
\end{definition}

\begin{definition}
For $f,g : X \to Y$, $p : x =_X y$, $\alpha : f \sim g$, there is a path
\[ \nathom_\alpha(p) : \ap{f}{p} \ct \alpha(y) = \alpha(x) \ct \ap{g}{p} \]
defined in the obvious way by induction on $p$ and referred to as the \emph{naturality} of the homotopy $\alpha$. Pictorially, we have
\begin{center}
\begin{tikzpicture}
\node (N0) at (1.5,0.85) {$\nathom_\alpha(p)$};
\node (N1) at (0,1.7) {$f(x)$};
\node (N2) at (0,0) {$f(y)$};
\node (N3) at (3,1.7) {$g(x)$};
\node (N4) at (3,0) {$g(y)$};
\draw[-] (N1) -- node[left]{$\ap{f}{p}$} (N2);
\draw[-] (N1) -- node[above]{$\alpha(x)$} (N3);
\draw[-] (N2) -- node[below]{$\alpha(y)$} (N4);
\draw[-] (N3) -- node[right]{$\ap{g}{p}$} (N4);
\end{tikzpicture}
\end{center}
\end{definition}

\noindent A crucial concept in HoTT is that of an equivalence between types. 
\begin{definition}
A map $f : A \to B$ is called an \emph{equivalence} if it has both a left and a right inverse:
\[\iseq{f} \defeq \big(\sm{g:B \to A} (g \circ f \sim \idfun{A})\big) \times \big(\sm{h:B \to A} (f \circ h \sim \idfun{B})\big) \]
We define \[(A \simeq B) \defeq \sm{f:A\to B}\iseq{f}\] and call $A$ and $B$ \emph{equivalent} if the above type is inhabited.
\end{definition}
We call $A$ and $B$ \emph{logically equivalent} if there are exist functions $f:A\to B$, $g:B \to A$. In practice, we often show that two types $A$ and $B$ are equivalent by first exhibiting the logical equivalence of $A$ and $B$ and then showing that the functions $f$ and $g$ compose to identity on both sides. In this case we refer to $f$ and $g$ as forming a \emph{quasi-equivalence} and say that $f$ and $g$ are \emph{quasi-inverses} of each other. A pair of quasi-inverses can always be turned into an equivalence.

Many ``diagram-like" operations on paths turn out to be equivalences. For instance:

\begin{itemize}
\item For any $u : a =_A b$, $v : b =_A d$, $w : a =_A c$, $z : c =_A d$, as in the diagram
\begin{center}
\begin{tikzpicture}
\node (N0) at (0,2) {$a$};
\node (N1) at (2,2) {$b$};
\node (N2) at (0,0) {$c$};
\node (N3) at (2,0) {$d$};
\draw[-] (N0) -- node[above]{$u$} (N1);
\draw[-] (N0) -- node[left]{$w$} (N2);
\draw[-] (N1) -- node[right]{$v$} (N3);
\draw[-] (N2) -- node[below]{$z$} (N3);
\end{tikzpicture}
\end{center}
we have functions
\begin{align*}
\I & : (u \ct v = w \ct z) \to (\opp{u} \ct w \ct z = v)\\
\I^{-1} & : (\opp{u} \ct w \ct z = v) \to (u \ct v = w \ct z)
\end{align*}
defined by path induction on $u$ and $z$, which form a quasi-equivalence.
\end{itemize}

Finally, we show how to construct paths in pair and function types. Given two pairs $c,d:A \times B$, we can easily construct a function \[\pair_{c,d} : (\id{c}{d}) \to(\id{\fst(c)}{\fst(d)}) \times (\id{\snd(c)}{\snd(d)}).\] We can show:
\begin{lemma}
The map $\pair_{c,d}$ is an equivalence for any $c,d:A\times B$.
\end{lemma}
\noindent We will denote the quasi-inverse of $\pair_{c,d}$ by $\pairpath_{c,d}$. For brevity we will often omit the subscripts. 

Analogously, given two functions $f,g : \prd{x:A} B(x)$, we can construct a function
\[\happly_{f,g} : (\id{f}{g}) \to (f \sim g)\]
Showing that this map is an equivalence (or even constructing a map in the opposite direction) is much harder, and is in fact among the chief consequences of the univalence axiom:
\begin{lemma}
The map $\happly_{f,g}$ is an equivalence for any $f,g:\prd{x:A}B(x)$.
\end{lemma}
\begin{proof}
See Ch.~4.9 of \cite{hott}.
\end{proof}
\noindent We will denote the quasi-inverse of $\happly_{f,g}$ by $\funext_{f,g}$.


\section{The circle $\Sn^1$ and the torus $T^2$}
The circle $\Sn^1$ is a higher inductive type generated by the constructors
\begin{alignat*}{4}
& \base & \; : \; & \Sn^1, \\
& \lloop & \; : \; & \base = \base.
\end{alignat*}
The recursion principle says that given a type $C : \type$ and terms
\begin{alignat*}{4}
& b & \; : \; & C, \\
& l & \; : \; & b = b
\end{alignat*}
there exists a recursor $f : \Sn^1 \to C$ for which $f(\base) \equiv b$ and $\ap{f}{\lloop} = l$. The induction principle says that given a family $E : \Sn^1 \to \type$ and terms
\begin{alignat*}{4}
& b & \; : \; & E(\base), \\
& l & \; : \; & \transfib{E}{\lloop}{b} = b
\end{alignat*}
there exists an inductor $f : \prd{x:\Sn^1} E(x)$ for which $f(\base) \equiv b$ and $\mapdepfunc{f}(\lloop) = l$.\bigskip

The torus $T^2$ is a higher inductive type generated by the constructors
\begin{alignat*}{4}
& \baseT & \; : \; & T^2, \\
& \pT & \; : \; & \baseT = \baseT, \\
& \qT & \; : \; & \baseT = \baseT, \\
& \tT & \; : \; & \pT \ct \qT = \qT \ct \pT
\end{alignat*}
as pictured below:
\begin{center}
\begin{tikzpicture}
\node (N0) at (0,2) {$\baseT$};
\node (N1) at (2,2) {$\baseT$};
\node (N2) at (0,0) {$\baseT$};
\node (N3) at (2,0) {$\baseT$};
\node at (1,1) {$\Downarrow \tT$};
\draw[-] (N0) -- node[above]{$\pT$} (N1);
\draw[-] (N0) -- node[left]{$\qT$} (N2);
\draw[-] (N1) -- node[right]{$\qT$} (N3);
\draw[-] (N2) -- node[below]{$\pT$} (N3);
\end{tikzpicture}
\end{center}
The recursion principle says that given a type $C : \type$ and terms
\begin{alignat*}{4}
& b' & \; : \; & C, \\
& p' & \; : \; & b' = b', \\
& q' & \; : \; & b' = b', \\
& t' & \; : \; & p' \ct q' = q' \ct p',
\end{alignat*}
there exists a recursor $f : T^2 \to C$ for which $f(\baseT) \equiv b'$ and there exist paths $\beta : \ap{f}{\pT} = p'$ and $\gamma : \ap{f}{\qT} = q'$ making the following diagram commute:
\begin{center}
\begin{tikzpicture}
\node (N0) at (0,2) {$\ap{f}{p\ct q}$};
\node (N1) at (4,2) {$\ap{f}{q\ct p}$};
\node (N2) at (0,0) {$\ap{f}{p}\ct\ap{f}{q}$};
\node (N3) at (4,0) {$\ap{f}{q}\ct\ap{f}{p}$};
\node (N4) at (0,-2) {$p' \ct q'$};
\node (N5) at (4,-2) {$q' \ct p'$};
\draw[-] (N0) -- node[above]{via $\tT$} (N1);
\draw[-] (N0) -- node[left]{} (N2);
\draw[-] (N1) -- node[above]{} (N3);
\draw[-] (N2) -- node[left]{via $\beta$, $\gamma$} (N4);
\draw[-] (N3) -- node[right]{via $\gamma$, $\beta$} (N5);
\draw[-] (N4) -- node[below]{$t'$} (N5);
\end{tikzpicture}
\end{center}
Here, each edge represents an equality between its vertices. Unlabeled edges stand for the ``obvious" equalities which follow from the basic properties of identity types, such as the path from $\ap{f}{p\ct q}$ to $\ap{f}{p}\ct\ap{f}{q}$. Edges labeled with, \emph{e.g.}, ``$\text{via} \; \beta, \gamma$" stand for an application of congruence: here $\beta$ is a path from $\ap{f}{p}$ to $p'$ and $\gamma$ is a path from $\ap{f}{q}$ to $q'$. Since path concatenation respects equality, combining $\beta$ and $\gamma$ in a straightforward fashion yields a path from $\ap{f}{p}\ct\ap{f}{q}$ to $p' \ct q'$. 

We note that there may be several natural ways how to implement, \emph{e.g.}, the congruence of path concatenation with respect to path equality: we can perform path induction on the first argument, on the second, or on both. For our purposes the exact definition is immaterial as they are all equal up to a higher path, which is why we only specify the arguments (in this case $\beta$ and $\gamma$). From now on, all paths and diagrams will be annotated in this style.

The induction principle for $\T^2$ is more complicated; it says that given a family $E : T^2 \to \type$, in order to get an inductor $f : \prd{x:T^2} E(x)$ we require terms
\begin{alignat*}{4}
& b' & \; : \; & E(\baseT) \\
& p' & \; : \; & \transfib{E}{\pT}{b'} = b' \\
& q' & \; : \; & \transfib{E}{\qT}{b'} = b' \\
& t' & \; : \; & \big(\ap{\alpha \mapsto \transfib{E}{\alpha}{b'}}{\tT}\big)^{-1} \ct \Big(\T_f(E,\pT,\qT,b') \ct \ap{\mathtt{trans}^E(\qT)}{p'} \ct q'\Big) = \\
& & & \;\;\;\; \;\;\;\; \T_f(E,\qT,\pT,b') \ct \ap{\mathtt{trans}^E(\pT)}{q'} \ct p'
\end{alignat*}
where for any family $E : T^2 \to \type$, paths $\alpha: x =_{T^2} y$, $\alpha' : y =_{T^2} z$ and point $u : E(x)$, the path 
\[\T_f(E,\alpha,\alpha',u) : \transfib{E}{\alpha \ct \alpha'}{u} = \transfib{E}{\alpha'}{\transfib{E}{\alpha}{u}}\] 
is obtained by path induction on $\alpha$ and $\alpha'$. 
The inductor $f$ then has the property that $f(b) \equiv b'$. Furthermore, there exist paths $\beta : \mapdepfunc{f}(\pT) = p'$ and $\gamma : \mapdepfunc{f}{(\qT)} = q'$ satisfying a higher coherence law, which we omit since we do not need it.

\section{Logical equivalence between $\Sn^1 \times \Sn^1$ and $T^2$}
\paragraph*{Left-to-right}
We define a function $f : \Sn^1 \to T^2$ by circle recursion, mapping $\base \mapsto \baseT$ and $\lloop \mapsto \pT$. Thus, we have a definitional equality $f(\base) \equiv \baseT$ and a path $\beta_f : \ap{f}{\lloop} = \pT$.

We define a function $F^\to : \Sn^1 \to \Sn^1 \to T^2$ again by circle recursion, mapping $\base \mapsto f$ and $\lloop \mapsto \funext(H)$, where $H : \prd{x:\Sn^1} f(x) = f(x)$ is defined by circle induction as follows. We map $\base$ to $\qT$ and $\lloop$ to the path
\begin{center}
\begin{tikzpicture}
\node (N0) at (0,3.4) {$\transfib{z \mapsto f(z) = f(z)}{\lloop}{\qT}$};
\node (N1) at (0,1.7) {$\opp{\ap{f}{\lloop}} \ct \qT \ct \ap{f}{\lloop}$};
\node (N2) at (0,0) {$\qT$};
\draw[-] (N0) -- node[right]{\footnotesize $\T_1(\lloop,\qT)$} (N1);
\draw[-] (N1) -- node[right]{\footnotesize $\I(\gamma)$} (N2);
\end{tikzpicture}
\end{center}
where for any $\alpha : x =_{\Sn^1} y$ and $u : f(x) = f(x)$, the path 
\[\T_1(\alpha,u) : \transfib{z \mapsto f(z) = f(z)}{\alpha}{u} = \opp{\ap{f}{\alpha}} \ct u \ct \ap{f}{\alpha} \]
is obtained by a straightforward path induction on $\alpha$, and $\gamma$ is the path
\begin{center}
\begin{tikzpicture}
\node (N0) at (0,4.5) {$\ap{f}{\lloop} \ct \qT$};
\node (N1) at (0,3) {$\pT \ct \qT$};
\node (N2) at (0,1.5) {$\qT \ct \pT$};
\node (N3) at (0,0) {$\qT \ct \ap{f}{\lloop}$};
\draw[-] (N0) -- node[right]{\footnotesize via $\beta_f$} (N1);
\draw[-] (N1) -- node[right]{\footnotesize $\tT$} (N2);
\draw[-] (N2) -- node[right]{\footnotesize via $\beta_f^{-1}$} (N3);
\end{tikzpicture}
\end{center}
Having defined a function $F^\to : \Sn^1 \to \Sn^1 \to T^2$, it is now straightforward to define its curried version $F : \Sn^1 \times \Sn^1 \to T^2$. We note that $F^\to(\base) \equiv f$, and in particular $F(\base,\base) \equiv \baseT$. Furthermore, we have a path $\beta_{F^\to} : \ap{F^\to}{\lloop} = \funext(H)$. Since $\happly$ and $\funext$ form a quasi-equivalence, we have a path
\[ \beta^*_{F^\to} : \happly(\ap{F^\to}{\lloop}) = H \]
The function $H$ is a homotopy between $f$ and $f$ such that $H(\base) \equiv \qT$ and the following diagram commutes:
\begin{center}
\begin{tikzpicture}
\node (Na) at (3,1) {$(1)$};
\node (N0) at (0,2) {$\ap{f}{\lloop} \ct \qT$};
\node (N1) at (0,0) {$\pT \ct \qT$};
\node (N2) at (6,2) {$\qT \ct \ap{f}{\lloop}$};
\node (N3) at (6,0) {$\qT \ct \pT$};
\draw[-] (N0) -- node[left]{\footnotesize via $\beta_f$} (N1);
\draw[-] (N1) -- node[below]{\footnotesize $\tT$} (N3);
\draw[-] (N0) -- node[above]{\footnotesize $\nathom_H(\lloop)$} (N2);
\draw[-] (N2) -- node[right]{\footnotesize via $\beta_f$} (N3);
\end{tikzpicture}
\end{center}
To show this, we note that for any $\alpha : x =_{\Sn^1} y$, applying $\I^{-1}$ to the path
\begin{center}
\begin{tikzpicture}
\node (N0) at (0,3.4) {$\opp{\ap{f}{\alpha}} \ct H(x) \ct \ap{f}{\alpha}$};
\node (N1) at (0,1.7) {$\transfib{z \mapsto f(z) = f(z)}{\alpha}{H(x)}$};
\node (N2) at (0,0) {$H(y)$};
\draw[-] (N0) -- node[right]{\footnotesize $\opp{\T_1(\alpha,H(x))}$} (N1);
\draw[-] (N1) -- node[right]{\footnotesize $\mapdep{H}{\alpha}$} (N2);
\end{tikzpicture}
\end{center}
yields $\nathom_H(\alpha)$: this follows by a path induction on $\alpha$ and a subsequent generalization and path induction on $H(x)$.
The second computation rule for $H$ tells us that \[\mapdep{H}{\lloop} = \T_1(\lloop,\qT) \ct \I(\gamma)\]
Thus
\[\nathom_H(\lloop) = \I^{-1}\big(\opp{\T_1(\lloop,\qT)} \ct \mapdep{H}{\lloop}\big) = \gamma \]
which proves the commutativity of $(1)$.

Finally, we note that for any $\alpha : x =_{T^2} x'$ and $\alpha' : y =_{T^2} y'$, we have path families
\begin{alignat*}{4}
& \mu_x(\alpha') & \;  : \; & \ap{F}{\pairpath(\refl{x},\alpha')} = \ap{F^\to(x)}{\alpha'}\\
& \nu_y(\alpha) & \; : \; & \ap{F}{\pairpath(\alpha,\refl{y})} = \happly(\ap{F^\to}{\alpha},y)
\end{alignat*}
defined by path induction on $\alpha'$ and $\alpha$ respectively.

\paragraph*{Right-to-left}
We define a function $G : T^2 \to \Sn^1 \times \Sn^1$ by torus recursion as follows. We map $\baseT \mapsto (\base,\base)$, $\pT \mapsto \pairpath(\refl{\base},\lloop)$, $\qT \mapsto \pairpath(\lloop,\refl{\base})$, and $\tT \mapsto \Phi_{\lloop,\lloop}$, where for any $\alpha : x =_{\Sn^1} x'$, $\alpha' : y =_{\Sn^1} y'$, the path
\[ \Phi_{\alpha,\alpha'} : \Big(\pairpath(\refl{x},\alpha') \ct \pairpath(\alpha, \refl{y'})\Big) = \Big(\pairpath(\alpha, \refl{y}) \ct \pairpath(\refl{x'},\alpha')\Big) \]
is defined by induction on $\alpha$ and $\alpha'$.

Then we have a definitional equality $G(\baseT) \equiv (\base,\base)$ and paths
\begin{align*}
& \beta^p_G : \ap{G}{\pT} = \pairpath(\refl{\base},\lloop)\\
& \beta^q_G : \ap{G}{\qT} =\pairpath(\lloop, \refl{\base})
\end{align*}
which make the following diagram commute:
\begin{center}
\begin{tikzpicture}
\node (Nc) at (4.5,0) {$(2)$};
\node (N0) at (0,1.7) {$\ap{G}{\pT\ct \qT}$};
\node (N1) at (9,1.7) {$\ap{G}{\qT\ct \pT}$};
\node (N2) at (0,0) {$\ap{G}{\pT}\ct\ap{G}{\qT}$};
\node (N3) at (9,0) {$\ap{G}{\qT}\ct\ap{G}{\pT}$};
\node (N4) at (0,-1.7) {$\pairpath(\refl{},\lloop) \ct \pairpath(\lloop, \refl{})$};
\node (N5) at (9,-1.7) {$\pairpath(\lloop, \refl{}) \ct \pairpath(\refl{},\lloop)$};
\draw[-] (N0) -- node[above]{\footnotesize via $\tT$} (N1);
\draw[-] (N0) -- node[left]{} (N2);
\draw[-] (N1) -- node[above]{} (N3);
\draw[-] (N2) -- node[left]{\footnotesize via $\beta^p_G$, $\beta^q_G$} (N4);
\draw[-] (N3) -- node[right]{\footnotesize via $\beta^q_G$, $\beta^p_G$} (N5);
\draw[-] (N4) -- node[below]{\footnotesize $\Phi_{\lloop,\lloop}$} (N5);
\end{tikzpicture}
\end{center}


\section{Equivalence between $\Sn^1 \times \Sn^1$ and $T^2$}
\paragraph*{Left-to-right}
We need to show that for any $x,y : \Sn^1$ we have $G(F(x,y)) = (x,y)$. We do this by circle induction on the first argument. We need a path family $\epsilon : \prd{y:\Sn^1} G(f(y)) = (\base,y)$. The definition of $\epsilon$ itself proceeds by circle induction: we map $\base$ to the path $\refl{(\base,\base)}$ and $\lloop$ to the path
\begin{center}
\begin{tikzpicture}
\node (N0) at (0,3.4) {$\transfib{z \mapsto G(f(z)) = (\base,z)}{\lloop}{\refl{(\base,\base)}}$};
\node (N1) at (0,1.7) {$\opp{\ap{G}{\ap{f}{\lloop}}} \ct \refl{(\base,\base)} \ct \pairpath(\refl{\base},\lloop)$};
\node (N2) at (0,0) {$\refl{(\base,\base)}$};
\draw[-] (N0) -- node[right]{\footnotesize $\T_2(\lloop,\refl{(\base,\base)})$} (N1);
\draw[-] (N1) -- node[right]{\footnotesize $\I(\delta)$} (N2);
\end{tikzpicture}
\end{center}
where for any $\alpha : x =_{\Sn^1} y$ and $u : G(f(x)) = (\base,x)$, the path \[\T_2(\alpha,u) : \transfib{z \mapsto G(f(z)) = (\base,z)}{\alpha}{u} = \opp{\ap{G}{\ap{f}{\alpha}}} \ct u \ct \pairpath(\refl{\base},\alpha) \]
is defined by path induction on $\alpha$ and $\delta$ is the path
\begin{center}
\begin{tikzpicture}
\node (N0) at (0,6) {$\ap{G}{\ap{f}{\lloop}} \ct \refl{(\base,\base)}$};
\node (N1) at (0,4.5) {$\ap{G}{\ap{f}{\lloop}}$};
\node (N2) at (0,3) {$\ap{G}{\pT}$};
\node (N3) at (0,1.5) {$\pairpath(\refl{\base},\lloop)$};
\node (N4) at (0,0) {$\refl{(\base,\base)} \ct \pairpath(\refl{\base},\lloop)$};
\draw[-] (N0) -- node[right]{} (N1);
\draw[-] (N1) -- node[right]{\footnotesize via $\beta_f$} (N2);
\draw[-] (N2) -- node[right]{\footnotesize $\beta^p_G$} (N3);
\draw[-] (N3) -- node[right]{} (N4);
\end{tikzpicture}
\end{center}
This finishes the definition of $\epsilon$. We now need to prove that \[\transfib{x \mapsto \Pi(y:\Sn^1) G(F(x,y)) = (x,y)}{\lloop}{\epsilon} = \epsilon\] 
By function extensionality, it suffices to show that for any $y : \Sn^1$ we have
\[\transfib{x \mapsto \Pi(z:\Sn^1) G(F(x,z)) = (x,z)}{\lloop}{\epsilon} \; y = \epsilon(y)\] 
Straightforward path induction shows that for any $\alpha : \base =_{\Sn^1} x$, we have\[\transfib{w \mapsto \Pi(z:\Sn^1) G(F(w,z)) = (w,z)}{\alpha}{\epsilon} \; y = \opp{\ap{G}{\happly(\ap{F^\to}{\alpha},y)}} \ct \epsilon(y) \ct \pairpath(\alpha,\refl{y}) \] 
It thus suffices to show that
\[ \ap{G}{\happly(\ap{F^\to}{\lloop},y)} \ct \epsilon(y) = \epsilon(y) \ct \pairpath(\lloop,\refl{y})\]
After simplifying the left endpoint using $\happly(\beta^*_{F^\to},y)$ it suffices to show that
\[ \ap{G}{H(y)} \ct \epsilon(y) = \epsilon(y) \ct \pairpath(\lloop,\refl{y}) \]
for any $y :\Sn^1$. We proceed yet again by circle induction. We map $\base$ to the path $\eta$ below:
\begin{center}
\begin{tikzpicture}
\node (N0) at (0,3) {$\ap{G}{\qT} \ct \refl{(\base,\base)}$};
\node (N1) at (0,1.5) {$\ap{G}{\qT}$};
\node (N2) at (0,0) {$\pairpath(\lloop,\refl{\base})$};
\node (N3) at (0,-1.5) {$\refl{(\base,\base)} \ct \pairpath(\lloop,\refl{\base})$};
\draw[-] (N0) -- node[right]{\footnotesize} (N1);
\draw[-] (N1) -- node[right]{\footnotesize $\beta^q_G$} (N2);
\draw[-] (N2) -- node[right]{\footnotesize} (N3);
\end{tikzpicture}
\end{center}
All that now remains to show is
\[ \transfib{y \mapsto \ap{G}{H(y)} \ct \epsilon(y) = \epsilon(y) \ct \pairpath(\lloop,\refl{y})}{\lloop}{\eta} = \eta \]
However, this follows at once from the fact that the circle $\Sn^1$, and hence the product $\Sn^1 \times \Sn^1$, is a 1-type (as shown \emph{e.g.}, by Licata and Shulman in \cite{licata_shulman}): this means that for any two points $x,y : \Sn^1 \times \Sn^1$, any two paths $\alpha,\alpha' : x = y$, and any two higher paths $\gamma, \gamma' : \alpha = \alpha'$, we necessarily have $\gamma = \gamma'$.  

\paragraph{Right-to-left}
We need to show that for any $x:T^2$ we have $F(G(x)) = x$. We use torus induction with $\baseT' \defeq \refl{\baseT}$. We let $p'$ be the path
\begin{center}
\begin{tikzpicture}
\node (N1) at (0,9.9)  {$\transfib{x \mapsto F(G(x))=x}{\pT}{\refl{\baseT}}$};
\node (N2) at (0,8.25) {$\opp{\ap{F}{\ap{G}{\pT}}} \ct \refl{\baseT} \ct \pT$};
\node (N3) at (0,6.6) {$\refl{\baseT}$};
\draw[-] (N1) -- node[right]{\footnotesize $\T_3(\pT,\refl{\baseT})$} (N2);
\draw[-] (N2) -- node[right]{\footnotesize $\I(\kappa_p)$} (N3);
\end{tikzpicture}
\end{center}
where for any $\alpha : x =_{T^2} y$ and $u : F(G(x))=x$, the path \[\T_3(\alpha,u) : \transfib{x \mapsto F(G(x))=x}{\alpha}{u} = \opp{\ap{F}{\ap{G}{\alpha}}} \ct u \ct \alpha \]
is defined by path induction on $\alpha$, and $\kappa_p$ is the path in Fig.~\ref{fig:1}.
\begin{figure}
\begin{center}
\begin{tikzpicture}
\node (N2) at (0,8.25) {$\ap{F}{\ap{G}{\pT}} \ct \refl{\baseT}$};
\node (N3) at (0,6.6)  {$\ap{F}{\ap{G}{\pT}}$};
\node (N4) at (0,4.95) {$\ap{F}{\pairpath(\refl{\base},\lloop)}$};
\node (N5) at (0,3.3)  {$\ap{f}{\lloop}$};
\node (N6) at (0,1.65) {$\pT$};
\node (N7) at (0,0) {$\refl{\baseT} \ct \pT$};
\draw[-] (N2) -- node[right]{\footnotesize} (N3);
\draw[-] (N3) -- node[right]{\footnotesize via $\beta^p_G$} (N4);
\draw[-] (N4) -- node[right]{\footnotesize $\mu_\base(\lloop)$} (N5);
\draw[-] (N5) -- node[right]{\footnotesize $\beta_f$} (N6);
\draw[-] (N6) -- node[right]{\footnotesize} (N7);
\node at (0,-1) {\textbf{a)} $\kappa_p$};

\node (N2') at (7,8.25) {$\ap{F}{\ap{G}{\qT}} \ct \refl{\baseT}$};
\node (N3') at (7,6.6)  {$\ap{F}{\ap{G}{\qT}}$};
\node (N4') at (7,4.95) {$\ap{F}{\pairpath(\lloop,\refl{\base})}$};
\node (N5') at (7,3.3)  {$\happly(\ap{F^\to}{\lloop},\base)$};
\node (N6') at (7,1.65) {$\qT$};
\node (N7') at (7,0)   {$\refl{\baseT} \ct \qT$};
\draw[-] (N2') -- node[right]{\footnotesize} (N3');
\draw[-] (N3') -- node[right]{\footnotesize via $\beta^q_G$} (N4');
\draw[-] (N4') -- node[right]{\footnotesize $\nu_\base(\lloop)$} (N5');
\draw[-] (N5') -- node[right]{\footnotesize $\happly(\beta^*_{F^\to},\base)$} (N6');
\draw[-] (N6') -- node[right]{\footnotesize } (N7');
\node at (7,-1) {\textbf{b)} $\kappa_q$};
\end{tikzpicture}
\end{center}
\caption{The paths $\kappa_p$ and $\kappa_q$}
\label{fig:1}
\end{figure}
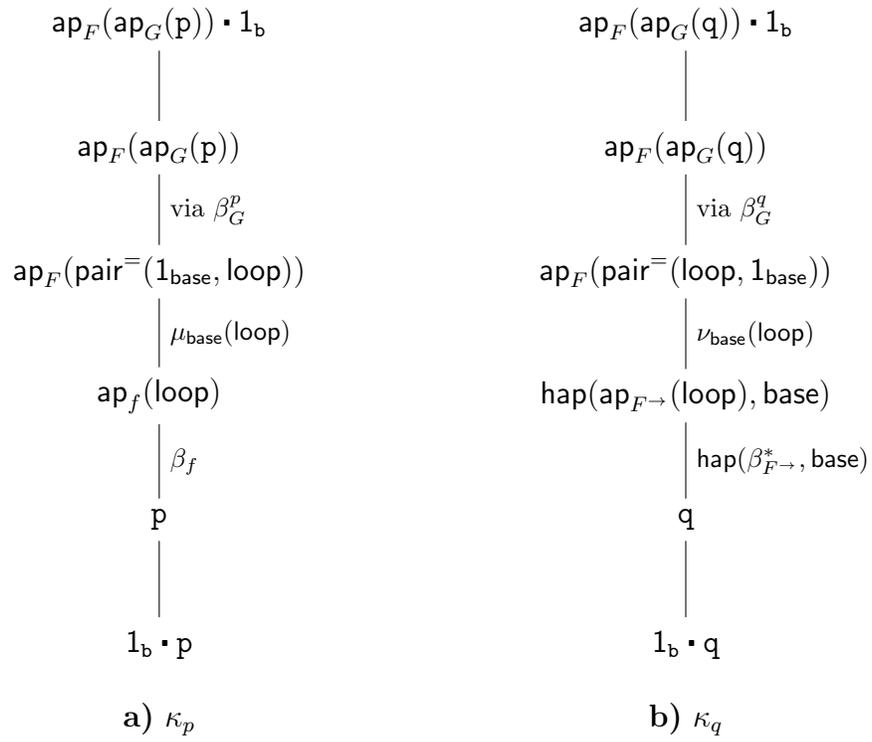
Similarly, let $q'$ be the path
\begin{center}
\begin{tikzpicture}
\node (N1) at (0,9.9)  {$\transfib{x \mapsto F(G(x))=x}{\qT}{\refl{\baseT}}$};
\node (N2) at (0,8.25) {$\opp{\ap{F}{\ap{G}{\qT}}} \ct \refl{\baseT} \ct \qT$};
\node (N3) at (0,6.6)   {$\refl{\baseT}$};
\draw[-] (N1) -- node[right]{\footnotesize $\T_3(\qT,\refl{\baseT})$} (N2);
\draw[-] (N2) -- node[right]{\footnotesize $\I(\kappa_q)$} (N3);
\end{tikzpicture}
\end{center}\newpage

All that remains now is to show that the following diagram commutes:
\begin{center}
\scalebox{0.95}{
\begin{tikzpicture}
\node (N1) at (0,9.9)  {\small $\transfib{x \mapsto F(G(x))=x}{\pT \ct \qT}{\refl{\baseT}}$};
\node (N2) at (0,8.25) {\small $\transfib{x \mapsto F(G(x))=x}{\qT}{\transfib{z \mapsto F(G(x))=x}{\pT}{\refl{\baseT}}}$};
\node (N3) at (0,6.6)  {\small $\transfib{x \mapsto F(G(x))=x}{\qT}{\refl{\baseT}}$};
\node (N4) at (4.25,4.95) {\small $\refl{\baseT}$};

\node (N5) at (8,9.9)  {\small $\transfib{x \mapsto F(G(x))=x}{\qT \ct \pT}{\refl{\baseT}}$};
\node (N6) at (8,8.25) {\small $\transfib{x \mapsto F(G(x))=x}{\pT}{\transfib{z \mapsto F^\times(G(z))=z}{\qT}{\refl{\baseT}}}$};
\node (N7) at (8,6.6)  {\small $\transfib{x \mapsto F(G(x))=x}{\pT}{\refl{\baseT}}$};

\draw[-] (N1) -- node[above]{\small via $\tT$} (N5);
\draw[-] (N1) -- node[right]{\footnotesize} (N2);
\draw[-] (N2) -- node[right]{\small via $p'$} (N3);
\draw[-] (N3) -- node[below]{\small $q'$} (N4);
\draw[-] (N5) -- node[right]{\footnotesize} (N6);
\draw[-] (N6) -- node[right]{\small via $q'$} (N7);
\draw[-] (N7) -- node[below]{\small $p'$} (N4);
\end{tikzpicture}}
\end{center}
\noindent We proceed in four steps.

\paragraph{Step 1}
For terms $\alpha_1 : x =_{T^2} y$, $\alpha_2 : y =_{T^2} z$, $\alpha'_1 : a =_{T^2} b$, $\alpha'_2 : b =_{T^2} c$,
$u_x : F(G(x))=a$, $u_y : F(G(y))=b$, $u_z : F(G(z))=c$, $\eta_1 : \ap{F}{\ap{G}{\alpha_1}} \ct u_y = u_x \ct \alpha'_1$, $\eta_2 : \ap{F}{\ap{G}{\alpha_2}} \ct u_z = u_y \ct \alpha'_2$,
let $\zeta(\alpha_1,\alpha_2,\alpha'_1,\alpha'_2,u_x,u_y,u_z,\eta_1,\eta_2)$ be the path in Fig.~\ref{fig:2}.
\begin{figure}
\begin{center}
\begin{tikzpicture}
\node (N1) at (0,10.2)  {$\ap{F}{\ap{G}{\alpha_1 \ct \alpha_2}} \ct u_z$};
\node (N2) at (0,8.5) {$\Big(\ap{F}{\ap{G}{\alpha_1}} \ct \ap{F}{\ap{G}{\alpha_2}}\Big) \ct u_z$};
\node (N3) at (0,6.8) {$\ap{F}{\ap{G}{\alpha_1}} \ct \Big(\ap{F}{\ap{G}{\alpha_2}} \ct u_z\Big)$};
\node (N4) at (0,5.1) {$\ap{F}{\ap{G}{\alpha_1}} \ct (u_y \ct \alpha_2')$};
\node (N5) at (0,3.4)  {$\Big(\ap{F}{\ap{G}{\alpha_1}} \ct u_y\Big) \ct \alpha_2'$};
\node (N6) at (0,1.7) {$(u_x \ct \alpha'_1) \ct \alpha'_2$};
\node (N7) at (0,0)  {$u_x \ct (\alpha'_1 \ct \alpha'_2)$};

\draw[-] (N1) -- node[right]{} (N2);
\draw[-] (N2) -- node[right]{\footnotesize} (N3);
\draw[-] (N3) -- node[right]{\small via $\eta_2$} (N4);
\draw[-] (N4) -- node[right]{\footnotesize} (N5);
\draw[-] (N5) -- node[right]{\small via $\eta_1$} (N6);
\draw[-] (N6) -- node[right]{\footnotesize} (N7);
\end{tikzpicture}
\end{center}
\caption{The path $\zeta(\alpha_1,\alpha_2,\alpha'_1,\alpha'_2,u_x,u_y,u_z,\eta_1,\eta_2)$}
\label{fig:2}
\end{figure}

Now for $\alpha_1 : x =_{T^2} y$, $\alpha_2 : y =_{T^2} z$, $u_x : F(G(x))=x$, $u_y : F(G(y))=y$, $u_z : F(G(z))=z$, $\eta_1 : \ap{F}{\ap{G}{\alpha_1}} \ct u_ y = u_x \ct \alpha_1$, $\eta_2 : \ap{F}{\ap{G}{\alpha_2}} \ct u_z = u_y \ct \alpha_2$, we claim the path
\begin{center}
\begin{tikzpicture}
\node (N1) at (0,5.1)  {$\transfib{w \mapsto F(G(w))=w}{\alpha_1 \ct \alpha_2}{u_x}$};
\node (N2) at (0,3.4) {$\transfib{w \mapsto F(G(w))=w}{\alpha_2}{\transfib{w \mapsto F(G(w))=w}{\alpha_1}{u_x}}$};
\node (N3) at (0,1.7) {$\transfib{w \mapsto F(G(w))=w}{\alpha_2}{u_y}$};
\node (N4) at (0,0) {$u_z$};

\draw[-] (N1) -- node[right]{\small} (N2);
\draw[-] (N2) -- node[right]{\small via $\T_3(\alpha_1,u_x) \ct \I(\eta_1)$} (N3);
\draw[-] (N3) -- node[right]{\small $\T_3(\alpha_2,u_y) \ct \I(\eta_2)$} (N4);
\end{tikzpicture}
\end{center}
is equal to the path
\begin{center}
\begin{tikzpicture}
\node (N1) at (0,3.4)  {$\transfib{w \mapsto F(G(w))=w}{\alpha_1 \ct \alpha_2}{u_x}$};
\node (N2) at (0,1.7) {$\opp{\ap{F}{\ap{G}{\alpha_1 \ct \alpha_2}}} \ct u_x \ct (\alpha_1 \ct \alpha_2)$};
\node (N3) at (0,0) {$u_z$};

\draw[-] (N1) -- node[right]{\small $\T_3(\alpha_1\ct \alpha_2,u_x)$} (N2);
\draw[-] (N2) -- node[right]{\small via $\I(\zeta(\alpha_1,\alpha_2,\alpha_1,\alpha_2,u_x,u_y,u_z,\eta_1,\eta_2))$} (N3);
\end{tikzpicture}
\end{center}

To show this, we proceed by path induction on $\alpha_1$ and $\alpha_2$. Hence we have to establish the claim for $\alpha_1 := \refl{x}$, $\alpha_2 := \refl{x}$, $u_x,u_y,u_z : F(G(x))=x$, and $\eta_1 : \refl{F(G(x))} \ct u_y = u_x \ct \refl{x}$, $\eta_2 : \refl{F(G(x))} \ct u_z = u_y \ct \refl{x}$. 

We note, however, that the types of $\eta_1,\eta_2$ are equivalent to $u_x = u_y$ and $u_y = u_z$ respectively. Hence it suffices to
show that given $u_x,u_y,u_z : F(G(x))=x$, $\eta_1' : u_x = u_y$, $\eta_2' : u_y = u_z$, we can establish the claim for the special case when $\eta_1$ and $\eta_2$ have been obtained from $\eta_1'$ and $\eta_2'$, respectively, by using the aforementioned equivalences.

But we can now perform path induction on $\eta_1'$ and $\eta_2'$, leaving us with $u_x : F(G(x)) = x$ and $\eta_1' := \refl{u_x}$, $\eta_2' := \refl{u_x}$. We finish the proof by generalizing the endpoints of $u_x$ and performing a final path induction. \bigskip

By what we have just shown, it suffices to prove that the following diagram commutes:

\begin{center}
\begin{tikzpicture}
\node (N1) at (0,4)  {\small $\transfib{w \mapsto F(G(w))=w}{\pT \ct \qT}{\refl{\baseT}}$};
\node (N2) at (0,2) {\small $\opp{\ap{F}{\ap{G}{\pT \ct \qT}}} \ct \refl{\baseT} \ct (\pT \ct \qT)$};
\node (N3) at (4,0) {\small $\refl{\baseT}$};

\node (N4) at (8,4)  {\small $\transfib{x \mapsto F(G(w))=w}{\qT \ct \pT}{\refl{\baseT}}$};
\node (N5) at (8,2) {\small $\opp{\ap{F}{\ap{G}{\qT \ct \pT}}} \ct \refl{\baseT} \ct (\qT \ct \pT)$};

\draw[-] (N1) -- node[above]{\small via $\tT$} (N4);
\draw[-] (N1) -- node[right]{\small $\T_3(\pT \ct \qT,\refl{\baseT})$} (N2);
\draw[-] (N4) -- node[right]{\small $\T_3(\qT \ct \pT,\refl{\baseT})$} (N5);
\draw[-] (N2) -- node[below]{\small $\I(\zeta(\pT,\qT,\pT,\qT,\refl{\baseT},\refl{\baseT},\refl{\baseT},\kappa_p,\kappa_q))\;\;\;\;\;\;\;\;\;\;\;\;\;\;\;\;\;\;\;\;\;\;\;\;\;\;\;\;\;\;\;\;\;\;\;\;\;\;\;\;\;\;\;\;\;\;\;\;\;\;\;\;\;\;$} (N3);
\draw[-] (N5) -- node[below]{\small $\;\;\;\;\;\;\;\;\;\;\;\;\;\;\;\;\;\;\;\;\;\;\;\;\;\;\;\;\;\;\;\;\;\;\;\;\;\;\;\;\;\;\;\;\;\;\;\;\;\;\;\;\;\;\I(\zeta(\qT,\pT,\qT,\pT,\refl{\baseT},\refl{\baseT},\refl{\baseT},\kappa_q,\kappa_p))$} (N3);
\end{tikzpicture}
\end{center}

\paragraph{Step 2}
We observe the following: given terms $\alpha,\alpha' : x =_{T^2} y$, $\theta : \alpha = \alpha'$, $u_x : F(G(x))=x$, $u_y : F(G(y))=y$, $\eta : \ap{F}{\ap{G}{\alpha}} \ct u_ y = u_x \ct \alpha$, and $\eta' : \ap{F}{\ap{G}{\alpha'}} \ct u_y = u_x \ct \alpha'$, the commutativity of the diagram
\begin{center}
\begin{tikzpicture}
\node (N1) at (0,4)  {\small $\transfib{w \mapsto F(G(w))=w}{\alpha}{u_x}$};
\node (N2) at (0,2) {\small $\opp{\ap{F}{\ap{G}{\alpha}}} \ct u_x \ct \alpha$};
\node (N3) at (4,0) {\small $u_y$};

\node (N4) at (8,4)  {\small $\transfib{w \mapsto F(G(w))=w}{\alpha'}{u_x}$};
\node (N5) at (8,2) {\small $\opp{\ap{F}{\ap{G}{\alpha'}}} \ct u_x \ct \alpha'$};

\draw[-] (N1) -- node[above]{\small via $\theta$} (N4);
\draw[-] (N1) -- node[right]{\small $\T_3(\alpha,u_x)$} (N2);
\draw[-] (N4) -- node[right]{\small $\T_3(\alpha',u_x)$} (N5);
\draw[-] (N2) -- node[below]{\small $\I(\eta)\;\;\;\;$} (N3);
\draw[-] (N5) -- node[below]{\small $\;\;\;\; \I(\eta')$} (N3);
\end{tikzpicture}
\end{center}
is equivalent to the commutativity of the diagram
\begin{center}
\begin{tikzpicture}
\node (N1) at (0,2) {\small $\ap{F}{\ap{G}{\alpha}} \ct u_y$};
\node (N2) at (0,0) {\small $u_x \ct \alpha$};
\node (N3) at (6,2) {\small $\ap{F}{\ap{G}{\alpha'}} \ct u_y$};
\node (N4) at (6,0) {\small $u_x \ct \alpha'$};

\draw[-] (N1) -- node[above]{\footnotesize via $\theta$} (N3);
\draw[-] (N1) -- node[left]{\footnotesize $\eta$} (N2);
\draw[-] (N2) -- node[below]{\footnotesize via $\theta$} (N4);
\draw[-] (N3) -- node[right]{\footnotesize $\eta'$} (N4);
\end{tikzpicture}
\end{center}

To show this, we proceed by path induction on $\theta$ and a subsequent path induction on $\alpha$. After simplifying it remains to prove that for $u_x,u_y : F(G(x))=x$, $\eta,\eta' : \refl{F(G(x))} \ct u_y = u_x \ct \refl{x}$, we have 
$(\I(\eta) = \I(\eta')) \simeq (\eta = \eta')$. But this follows since $\I$ is an equivalence. \bigskip

By what we have just shown, it suffices to prove that the following diagram commutes:
\begin{center}
\begin{tikzpicture}
\node (N1) at (0,2) {\small $\ap{F}{\ap{G}{\pT \ct \qT}} \ct \refl{\baseT}$};
\node (N2) at (0,0) {\small $\refl{\baseT} \ct (\pT \ct \qT)$};
\node (N3) at (5.5,2) {\small $\ap{F}{\ap{G}{\qT \ct \pT}} \ct \refl{\baseT}$};
\node (N4) at (5.5,0) {\small $\refl{\baseT} \ct (\qT \ct \pT)$};

\draw[-] (N1) -- node[above]{\footnotesize via $\tT$} (N3);
\draw[-] (N1) -- node[left]{\small $\zeta(\pT,\qT,\pT,\qT,\refl{\baseT},\refl{\baseT},\refl{\baseT},\kappa_p,\kappa_q)$} (N2);
\draw[-] (N2) -- node[below]{\footnotesize via $\tT$} (N4);
\draw[-] (N3) -- node[right]{\small $\zeta(\qT,\pT,\qT,\pT,\refl{\baseT},\refl{\baseT},\refl{\baseT},\kappa_q,\kappa_p)$} (N4);
\end{tikzpicture}
\end{center}

\paragraph{Step 3}
We observe the following: for $k \in \{1,2 \}$ and $x_1,x_2,x_3 : T^2$ let terms $\alpha^1_k : x_k = x_{k+1}$; $\alpha^2_k : G(x_k) = G(x_{k+1})$; $\alpha^3_k, \alpha^4_k : F(G(x_k)) = F(G(x_{k+1}))$; $\iota^1_k : \ap{G}{\alpha^1_k} = \alpha^2_k$; $\iota_k^2 : \ap{F}{\alpha^2_k} = \alpha^3_k$; $\iota^3_k : \alpha^3_k = \alpha^4_k$ be given. Then the path
\begin{center}
\begin{tikzpicture}
\node (N1) at (0,10.2)  {$\ap{F}{\ap{G}{\alpha^1_1 \ct \alpha_2^1}} \ct \refl{F(G(x_3))}$};
\node (N2) at (0,8.5)  {$\ap{F}{\ap{G}{\alpha^1_1 \ct \alpha^1_2}}$};
\node (N3) at (0,6.8) {$\mathtt{ap}_F\Big(\ap{G}{\alpha^1_1} \ct \ap{G}{\alpha^1_2}\Big)$};
\node (N4) at (0,5.1) {$\ap{F}{\alpha^2_1 \ct \alpha^2_2}$};
\node (N5) at (0,3.4)  {$\ap{F}{\alpha^2_1} \ct \ap{F}{\alpha^2_2}$};
\node (N6) at (0,1.7) {$\alpha_1^3 \ct \alpha_2^3$};
\node (N7) at (0,0)  {$\alpha^4_1 \ct \alpha^4_2$};
\node (N8) at (0,-1.7)  {$\refl{F(G(x_1))} \ct (\alpha^4_1 \ct \alpha^4_2)$};

\draw[-] (N1) -- node[right]{} (N2);
\draw[-] (N2) -- node[right]{\footnotesize} (N3);
\draw[-] (N3) -- node[right]{\small via $\iota^1_1, \iota^1_2$} (N4);
\draw[-] (N4) -- node[right]{\footnotesize} (N5);
\draw[-] (N5) -- node[right]{\small via $\iota^2_1, \iota^2_2$} (N6);
\draw[-] (N6) -- node[right]{\small via $\iota^3_1, \iota^3_2$} (N7);
\draw[-] (N7) -- node[right]{\footnotesize} (N8);
\end{tikzpicture}
\end{center}
is equal to the path $\zeta\Big(\alpha^1_1,\alpha^1_2,\alpha^4_1,\alpha^4_2,\refl{F(G(x_1))},\refl{F(G(x_2))},\refl{F(G(x_3))},\eta_1,\eta_2\Big)$ where $\eta_k$ is the path
\begin{center}
\begin{tikzpicture}
\node (N2) at (0,8.25) {$\ap{F}{\ap{G}{\alpha^1_k}} \ct \refl{F(G(x_{k+1}))}$};
\node (N3) at (0,6.6)  {$\ap{F}{\ap{G}{\alpha^1_k}}$};
\node (N4) at (0,4.95) {$\ap{F}{\alpha^2_k}$};
\node (N5) at (0,3.3)  {$\alpha^3_k$};
\node (N6) at (0,1.65) {$\alpha^4_k$};
\node (N7) at (0,0)   {$\refl{F(G(x_k))} \ct \alpha^4_k$};
\draw[-] (N2) -- node[right]{\footnotesize} (N3);
\draw[-] (N3) -- node[right]{\footnotesize via $\iota^1_k$} (N4);
\draw[-] (N4) -- node[right]{\footnotesize $\iota^2_k$} (N5);
\draw[-] (N5) -- node[right]{\footnotesize $\iota^3_k$} (N6);
\draw[-] (N6) -- node[right]{\footnotesize } (N7);
\end{tikzpicture}
\end{center}

To show this, we proceed by path induction (with one endpoint fixed) on $\iota^1_k,\iota^2_k,\iota^3_k$ and a subsequent path induction on $\alpha^1_k$.\bigskip

By what we have just shown, it suffices to prove that the outer rectangle in the following diagram commutes:

\begin{center}
\scalebox{0.89}{
\begin{tikzpicture}
\node at (4.8,14.7) {A};
\node at (4.8,11.5) {B};
\node at (4.8,7) {C};
\node at (4.8,3.25) {D};
\node at (4.8,0.8) {E};

\node (N0) at (0,15.5) {\small $\ap{F}{\ap{G}{\pT \ct \qT}} \ct \refl{}$};
\node (N1) at (0,13.5) {\small $\ap{F}{\ap{G}{\pT \ct \qT}}$};
\node (N2) at (0,11.5) {\small $\mathtt{ap}_F\Big(\ap{G}{\pT} \ct \ap{G}{\qT}\Big)$};
\node (N3) at (0,9) {\small $\mathtt{ap}_F\Big(\pairpath(\refl{},\lloop) \ct \pairpath(\lloop,\refl{})\Big)$};
\node (N4) at (0,7) {\small $\ap{F}{\pairpath(\refl{},\lloop)} \ct \ap{F}{\pairpath(\lloop,\refl{})}$};
\node (N5) at (0,4.5) {\small $\ap{f}{\lloop} \ct \happly(\ap{F^\to}{\lloop},\base)$};
\node (N6) at (0,2) {\small $\pT \ct \qT$};
\node (N7) at (0,0)    {\small $\refl{} \ct (\pT \ct \qT)$};

\node (N10) at (9.6,15.5) {\small $\ap{F}{\ap{G}{\qT \ct \pT}} \ct \refl{}$};
\node (N11) at (9.6,13.5) {\small $\ap{F}{\ap{G}{\qT \ct \pT}}$};
\node (N12) at (9.6,11.5) {\small $\mathtt{ap}_F\Big(\ap{G}{\qT} \ct \ap{G}{\pT}\Big)$};
\node (N13) at (9.6,9) {\small $\mathtt{ap}_F\Big(\pairpath(\lloop,\refl{}) \ct \pairpath(\refl{},\lloop)\Big)$};
\node (N14) at (9.6,7) {\small $\ap{F}{\pairpath(\lloop,\refl{})} \ct \ap{F}{\pairpath(\refl{},\lloop)}$};
\node (N15) at (9.6,4.5) {\small $\happly(\ap{F^\to}{\lloop},\base) \ct \ap{f}{\lloop}$};
\node (N16) at (9.6,2) {\small $\qT \ct \pT$};
\node (N17) at (9.6,0)    {\small $\refl{} \ct (\qT \ct \pT)$};

\draw[-] (N0) -- node[right]{\scriptsize} (N1);
\draw[-] (N1) -- node[right]{\scriptsize} (N2);
\draw[-] (N2) -- node[right]{\small via $\beta^p_G$, $\beta^q_G$} (N3);
\draw[-] (N3) -- node[right]{\scriptsize} (N4);
\draw[-] (N4) -- node[right]{\small via $\mu_\base(\lloop)$, $\nu_\base(\lloop)$} (N5);
\draw[-] (N5) -- node[right]{\small via $\beta_f, \happly(\beta^*_{F^\to},\base)$} (N6);
\draw[-] (N6) -- node[right]{\scriptsize} (N7);
\draw[-] (N10) -- node[right]{\scriptsize} (N11);
\draw[-] (N11) -- node[right]{\scriptsize} (N12);
\draw[-] (N12) -- node[left]{\small  via $\beta^q_G$, $\beta^p_G$} (N13);
\draw[-] (N13) -- node[right]{\scriptsize} (N14);
\draw[-] (N14) -- node[left]{\small via $\nu_\base(\lloop)$, $\mu_\base(\lloop)$} (N15);
\draw[-] (N15) -- node[left]{\small via $\happly(\beta^*_{F^\to},\base),\beta_f$} (N16);
\draw[-] (N16) -- node[right]{\scriptsize} (N17);
\draw[-] (N0) -- node[above]{\small via $\tT$} (N10);
\draw[-] (N1) -- node[above]{\small via $\tT$} (N11);
\draw[-] (N7) -- node[below]{\small via $\tT$} (N17);
\draw[-] (N3) -- node[above]{\small via $\Phi_{\lloop,\lloop}$} (N13);
\draw[-] (N5) -- node[below]{\small $\nathom_{\happly(\ap{F^\to}{\lloop})}(\lloop)$} (N15);
\draw[-] (N6) -- node[below]{\small $\tT$} (N16);
\end{tikzpicture}}
\end{center}

\paragraph{Step 4}
It suffices to prove that each of the inner rectangles commutes. Rectangles A and E commute obviously. Rectangle B is just diagram (2)
transported along $\mathtt{ap}_F$, and hence commutes. Rectangle C commutes by the following generalization: for any $\alpha : x =_{\Sn^1} y$, the diagram below commutes by path induction on $\alpha$:
\begin{center}
\scalebox{0.90}{
\begin{tikzpicture}

\node (N3) at (0,4) {\small $\mathtt{ap}_F\Big(\pairpath(\refl{x},\alpha) \ct \pairpath(\alpha,\refl{y})\Big)$};
\node (N4) at (0,2) {\small $\ap{F}{\pairpath(\refl{x},\alpha)} \ct \ap{F}{\pairpath(\alpha,\refl{y})}$};
\node (N5) at (0,0) {\small $\ap{F^\to(x)}{\alpha} \ct \happly(\ap{F^\to}{\alpha},y)$};

\node (N13) at (9.6,4) {\small $\mathtt{ap}_F\Big(\pairpath(\alpha,\refl{x}) \ct \pairpath(\refl{y},\alpha)\Big)$};
\node (N14) at (9.6,2) {\small $\ap{F}{\pairpath(\alpha,\refl{x})} \ct \ap{F}{\pairpath(\refl{y},\alpha)}$};
\node (N15) at (9.6,0) {\small $\happly(\ap{F^\to}{\alpha},x) \ct \ap{F^\to(y)}{\alpha}$};

\draw[-] (N3) -- node[right]{\scriptsize} (N4);
\draw[-] (N4) -- node[right]{\small via $\mu_x(\alpha)$, $\nu_y(\alpha)$} (N5);
\draw[-] (N13) -- node[right]{\scriptsize} (N14);
\draw[-] (N14) -- node[left]{\small via $\nu_x(\alpha)$, $\mu_y(\alpha)$} (N15);
\draw[-] (N3) -- node[above]{\small via $\Phi_{\alpha,\alpha}$} (N13);
\draw[-] (N5) -- node[below]{\small $\nathom_{\happly(\ap{F^\to}{\alpha})}(\alpha)$} (N15);
\end{tikzpicture}}
\end{center}
It remains to show that rectangle D commutes. Consider the following diagram:
\begin{center}
\scalebox{0.933}{
\begin{tikzpicture}
\node at (4.8,3.55) {D$_1$};
\node at (4.8,1.05) {D$_2$};
\node (N5) at (0,5) {\small $\ap{f}{\lloop} \ct \happly(\ap{F^\to}{\lloop},\base)$};
\node (N6) at (0,2.5) {\small $\ap{f}{\lloop} \ct \qT$};
\node (N7) at (0,0) {\small $\pT \ct \qT$};

\node (N15) at (9.6,5) {\small $\happly(\ap{F^\to}{\lloop},\base) \ct \ap{f}{\lloop}$};
\node (N16) at (9.6,2.5) {\small $\qT \ct \ap{f}{\lloop}$};
\node (N17) at (9.6,0) {\small $\qT \ct \pT$};

\draw[-] (N5) -- node[right]{\small via $\happly(\beta^*_{F^\to},\base)$} (N6);
\draw[-] (N6) -- node[right]{\small via $\beta_f$} (N7);
\draw[-] (N15) -- node[left]{\small via $\happly(\beta^*_{F^\to},\base)$} (N16);
\draw[-] (N16) -- node[left]{\small via $\beta_f$} (N17);
\draw[-] (N7) -- node[below]{\small $\tT$} (N17);
\draw[-] (N5) -- node[below]{\small $\nathom_{\happly(\ap{F^\to}{\lloop})}(\lloop)$} (N15);
\draw[-] (N6) -- node[below]{\small $\nathom_H(\lloop)$} (N16);
\end{tikzpicture}}
\end{center}
Commutativity of the outer rectangle clearly implies the commutativity of D. It thus remains to show that D$_1$ and D$_2$ commute. The rectangle D$_2$ is precisely diagram (1), which commutes. Rectangle D$_1$ commutes by the following generalization: let $\gamma : h_1 =_{f \sim f} h_2$ and $\alpha : x =_{\Sn^1} y$ be given. Then the following diagram commutes by path induction on $\gamma$ and $\alpha$:
\begin{center}
\begin{tikzpicture}
\node (N5) at (0,2.5) {$\ap{f}{\alpha} \ct h_1(y)$};
\node (N6) at (0,0) {$\ap{f}{\alpha} \ct h_2(y)$};
\node (N15) at (6,2.5) {$h_1(x) \ct \ap{f}{\alpha}$};
\node (N16) at (6,0) {$h_2(x) \ct \ap{f}{\alpha}$};

\draw[-] (N5) -- node[left]{via $\happly(\gamma,y)$} (N6);
\draw[-] (N15) -- node[right]{via $\happly(\gamma,x)$} (N16);
\draw[-] (N5) -- node[above]{$\nathom_{h_1}(\alpha)$} (N15);
\draw[-] (N6) -- node[below]{$\nathom_{h_2}(\alpha)$} (N16);
\end{tikzpicture}
\end{center}
This finishes the proof.


\section{Conclusion}

We have presented a homotopy-type theoretic proof that the torus $T^2$ is equivalent to the product of two circles $\Sn^1 \times \Sn^1$. To compare the proof described here to the one given by Licata and Brunerie in \cite{licata_brunerie}, we first note that the definitions of the back-and-forth functions between $T^2$ and $\Sn^1 \times \Sn^1$ are exactly the same. When proving that the functions compose to the identity on $\Sn^1 \times \Sn^1$ we used the fact that the circle $\Sn^1$ is a 1-type. This simplification is not used by Licata and Brunerie; the lines \emph{75-76, 82-86} in \cite{licata_agda1} comprise the path algebra which would be avoided by the aforementioned simplification. On the other hand, in this fashion Agda is able to automatically infer the terms $\mathsf{loop1}$-$\mathsf{case}$ and $\mathsf{loop2}$-$\mathsf{case}$, which in our notation correspond to the paths $\eta$ and $\delta$ respectively (of course a paper proof offers no such opportunity).

Similarly, when proving that the functions compose to the identity on $\T^2$, the terms $\mathsf{p}$-$\mathsf{case}$ and $\mathsf{q}$-$\mathsf{case}$, which in our proof correspond to the paths $\kappa_q$ and $\kappa_p$, are inferred automatically. Steps 1 and 2 of our proof roughly correspond to lines \emph{403-441} in \cite{licata_agda2} and \emph{51-57} in \cite{licata_agda1}; in both proofs, the purpose of these steps is to mediate between a diagram involving transports (a ``square-over") and an equivalent diagram which does not (a ``cube"). Steps 3 and 4 then roughly correspond to lines \emph{60-67} in \cite{licata_agda1}; the commuting diagrams (or ``cubes") established in Step 4 are composed together, using a reordering of operations that is justified by Step 3.

\section*{Acknowledgement}
The author would like to thank her advisors, Profs. Steve Awodey and Frank Pfenning, for their help.

\bibliographystyle{plain}
\bibliography{references}

\end{document}